\documentclass[12pt]{article}

\usepackage{amssymb,amsmath,amsfonts,eurosym,geometry,ulem,graphicx,caption,color,setspace,sectsty,comment,footmisc,caption,natbib,pdflscape,subfigure,hyperref}

\usepackage{sgame} 
\usepackage{pgf, tikz} 
\usetikzlibrary{arrows, automata} 
\usepackage{ulem} 
\newtheorem{proposition}{Proposition}
\newenvironment{proof}[1][Proof]{\noindent\textbf{#1.} }{\ \rule{0.5em}{0.5em}}

\usepackage{natbib}
\begin{document}
\begin{titlepage}
\title{Limited Cognitive Abilities and Dominance Hierarchies\thanks{We sincerely thank the editor F.J.A. Jacobs and two anonymous reviewers for their suggestions that help to greatly improve the paper. We also thank Jonathan Newton for his comments.}}
\author{Hanyuan Huang\thanks{Department of Economics, University of Oregon, hanyuan@uoregon.edu} \and Jiabin Wu\thanks{Department of Economics, University of Oregon, jwu5@uoregon.edu}}
\date{\today}
\maketitle
\begin{abstract}
\noindent We propose a novel model to explain the mechanisms underlying dominance hierarchical structures. Guided by a predetermined social convention, agents with limited cognitive abilities optimize their strategies in a Hawk-Dove game. We find that several commonly observed hierarchical structures in nature such as linear hierarchy and despotism, emerge as the total fitness-maximizing social structures given different levels of cognitive abilities. \\
\vspace{0in}
\vspace{0in}\\
\\
\noindent\textbf{Declarations:} None.\\
\bigskip
\end{abstract}
\setcounter{page}{0}
\thispagestyle{empty}
\end{titlepage}
\pagebreak \newpage

\doublespacing

\section{Introduction} \label{sec:Introductio}

Dominance hierarchy is a social hierarchical structure in which a ranking system among the agents in a population can be induced based on their interactions. Introduced by \cite{Schjelderup-Ebbe1935} in describing the social organization of chickens, dominance hierarchies have been found to be very common as a function of regulating animal societies, especially in situations where there are potential costs and risks of conflict during interactions. More specifically, in a society with dominance hierarchy, pairwise interactions are regulated by a ranking system, where the higher-ranked agent in a pair acts dominantly, and the lower-ranked agent acts submissively. Linear hierarchy, known as pecking order, is a common social structure in various species including sheep, birds and crayfish (\cite{Addison1980}; \cite{Barkan1986}; \cite{Goessmann2000}; \cite{Hausfater1982}; \cite{Heinze1990}; \cite{Nelissen1985}; \cite{Savin-Williams1980}; \cite{Vannini1971}; \cite{Wang2011}). In this hierarchy, every agent is dominated by the higher-ranked members and in turn dominates the lower-ranked agents. 
Other nonlinear hierarchical structures have also been observed in nature. A typical nonlinear hierarchy is despotism, which can be found in hamsters, gorillas, and African wild dogs (\cite{Alcock2013}). It is a social order in which one agent dominates all others, with no dominance relations among subordinates. In addition, more complex nonlinear hierarchical social structures have been observed in dolphins, chimpanzees, baboons, and macaques (\cite{Holekamp1993}; \cite{Kummer1984}; \cite{Surbeck2011}). It is worth noting that if a hierarchical structure is not linear, then it must have at least one of the following properties (\cite{DeVries1995}): (1) there are two agents with equal hierarchical status, that is, they behave in an equal manner upon interaction; (2) there are two agents with an unknown or undefined relationship; and (3) there is a non-transitive relationship in a triad (A dominates B, B dominates C, and C dominates A). The commonality of the linear hierarchical structure in nature suggests that transitivity may be a desired property of dominance hierarchies (\cite{Appleby1983}).

Many studies have been conducted to explain how social hierarchies can be achieved and sustained. Two major approaches have been proposed in the literature. The first approach proposes that the social hierarchical structure is an external attribute and an expression of intrinsic physical or physiological differences among the agents, which can be directly (e.g., greater fighting abilities and larger bodies) or indirectly (e.g., better reproductive abilities and higher social status) related to the dominance behavior. This is also known as the prior attribute hypothesis (\cite{Drews1993}; \cite{Chase2002}). However, evidence has shown that it is difficult to predict the outcomes of dominance encounters for animals in small groups using these differences (\cite{Chase2011}). The second approach suggests that social hierarchies result from the dynamics of social interactions. The most representative case is the winner-and-loser effect, where individuals who win (lose) a contest have higher (lower) probability winning the next contest (\cite{Dugatkin2004}; \cite{Kura2016}; \cite{Goessmann2000}). Although it is the most representative explanation for hierarchy formation, the theory of winner-and-loser effect has been criticized for its arbitrary set-up values and lack of independence with regard to personality traits (\cite{Favati2017a}).

The two existing approaches are not mutually exclusive (\cite{Chase2011}), yet they have rarely been studied jointly (\cite{Favati2017a}). The aim of this research is to develop a model to explain some typical social hierarchical structures considering the pre-existed differences in the individual characteristics of the agents and the social interactions among them. Our model is built on the classical Hawk-Dove (HD) game. A population of agents is randomly matched to play the HD game, and each agent carries a unique identity. Thus, agents are able to condition their strategies on their opponents' identities. We then assume that the identities of the agents can be ranked linearly to capture the external differences in attributes, such as their biological characteristics. We also introduce a social convention to the game that provides ``suggestions" to the agents on their actions in the HD game according to their relative ranks. The purpose of this social convention can be seen as nature acting like a principal, who aims to maximize the fitness of a population of agents by indirectly influencing their behavior (\cite{Binmore1994}). 

Another important feature of the model is that we assume that the agents can only memorize some agents' identities, but not the rest. 
The restriction imposed on the agents' memories can be viewed as a case of bounded rationality framed as a costly computation (\cite{Halpern2015}). It captures the fact that it usually requires a certain level of cognitive ability for the agents to understand the characteristics of their opponents, to realize the useful information revealed by those characteristics, and to act accordingly. This limited ability to acquire others' identities, also referred as memory size, helps the agents in choosing their strategies, which include whose identities they choose to memorize and their corresponding actions in the HD game. The agents can condition their actions in the HD game on their opponents' identities only when the identities of these opponents are memorized.

We analyze the hierarchical social structures that emerge as equilibria in the model and consider those that follow the suggestion of the social convention and maximize the total fitness of the population. We find that different hierarchical social structures, including linear hierarchy and despotism, maximize total fitness in populations with different levels of cognitive abilities. Specifically, when the memory size of the agents is sufficiently large, the linear hierarchy is optimal. When the memory size is singular, despotism is optimal for small populations. We also conjecture that when cognitive ability is at a medium level, a three-layer dominance hierarchy structure, that divides the population into three classes, may be optimal. We confirm this conjecture through simulation. Hence, our model provides a mechanism that links cognitive ability with social hierarchy. 


A closely related paper by \cite{DoiNakamaru2018} studies the coevolution of transitive inference and memory capacity in the Hawk–Dove game. They find that when the cost of fighting is low, transitive inference with limited memory capacity has an evolutionary advantage because the agents can avoid costly fights via prompt formation of the dominance hierarchy which does not necessarily reflect the actual rank of the agents' resource-holding potential\footnote{See also \cite{NakamaruSasaki2003} for a study on the evolution of transitive inference.} While both we and \cite{DoiNakamaru2018} consider the dominance hierarchy and limited memory capacity, their approach is different from ours in several ways. First, in the model of \cite{DoiNakamaru2018}, agents engage in repeated interactions with one another, and their memories allow them to count a certain number of past wins and losses of the contests between two agents, which helps them determine the ranking of the two agents in terms of resource-holding potential. Instead, we consider that agents memorize the identities of some other agents. Second, \cite{DoiNakamaru2018} consider agents with different resource-holding potential, whereas the agents in our model are identical except for their identities. Hence, the agents in our model do not need to access who are stronger (weaker) than them.  Third, \cite{DoiNakamaru2018} examine the evolutionary stability of different combinations of inference procedures and memory capacity. In contrast, we investigate what equilibrium social structure can arise given different memory capacities and find those that maximize population fitness.

The remainder of this paper is organized as follows. In Section 2, the proposed model is presented. Section 3 analyzes the equilibria of the game from our model under various levels of limited cognitive abilities. Section 4 discusses a possible relaxation of restrictions imposed on the model. Section 5 consists of the conclusion.

\section{The Model} \label{sec:Model}

\subsection{The Hawk-Dove game}
Consider a population of $N$ agents who are randomly matched in pairs to play the Hawk-Dove (HD) game. The HD game has been widely used to model pairwise interactions where individuals contest a beneficial resource with a possibility of an escalated fight at a large cost, which constructs a simple situation where players have a choice to either being harsh (play Hawk) or soft (play Dove) on their opponents. The earliest illustration of the HD game is presented by \cite{Smith1973} in their analysis of animal behavioral strategies in contest situations. In this paper, we adopt the game form provided by \cite{Smith1976} as shown in Figure 1. 

 In this game, $V$ is the value of the contested resource and $C$ is the cost of an escalated fight. It is assumed that the value of the resource is less than the cost of a fight, that is, $C > V > 0$. Players split the beneficial resource equally if they both play Dove. They equally split the difference between the resource and fighting cost if they both play Hawk. The player who plays Hawk exclusively wins the resource if the other player plays Dove.
 

\begin{figure}[!htbp] 
\centering
	\caption{The Hawk-Dove Game \label{prisoners}}  
\begin{game}{2}{2}[Player 1][Player 2]
	    &  Hawk      &  Dove     \\
	 Hawk  &  $(V-C)/2, (V-C)/2$ & $V, 0$  \\
	 Dove  &  $0, V$ & $ V/2, V/2$\\
\end{game}
\end{figure}

The HD game carries three Nash equilibria (NEs), with two in pure strategy, $(Hawk, Dove)$ and $(Dove, Hawk)$, and one in mixed strategy, where each player plays \textit{Hawk} with a probability of $V/C$. The expected fitness from playing the mixed NE strategy is $(1-V/C)(V/2)$. 


\subsection{Identity and Social convention}

We assume that each agent is assigned with an identity. Identity plays the role of a name tag and it is unique for each agent. With this information, it is possible for the agents to condition their strategies on their opponents' identities. The difference in identities between two agents can be interpreted on the basis of differences in some of their biological characteristics, such as body size and reproductive ability, whereas in a social interpretation, it can be a label (e.g., representing different social classes in human societies) attached to the agents. 


A linear social rank over the agents' identities is assumed. With this assumption, identities can be written as numbers such that their values reflect the relative ranks. Without loss of generality, we say agents with smaller-valued identities are ranked higher.

There is also a social convention that gives favor to those who rank higher. This is done by imposing a suggestive rule on the agents, regulating their actions in the HD game in accordance with the predetermined social rank. Specifically, it suggests that agents play $Hawk$ against opponents with lower ranks, and play $Dove$ against opponents with higher ranks. Since everyone is assumed to have a unique identity, if the social convention is followed, in any equilibrium, one player plays $Hawk$ and the other plays $Dove$ in each pairwise interaction.

Importantly, a linear social rank is assumed for the purpose of designing the social convention. The numerical values of the numbers are meaningless to the agents, and they do not need to know the social rank. For example, if an agent’s assigned number is 5 and it knows the identity of the agent whose assigned number is 1, then given the suggestion of the social convention, the number 5 agent should play Dove to the number 1 agent. 

A pertinent question is, who design the social convention? We assume nature acts as the principal, trying to maximize the fitness of the population by indirectly influencing agents’ behavior \citep{Binmore1994}. The social convention based on the linear social rank can guide agents to avoid playing $(Hawk, Hawk)$, which is the only strategy profile that generates fitness loss ($C$) on the scale of the full population in the HD game (Figure 1).

\subsection{Memory}

In a complete information environment, each agent's identity is common knowledge to all agents; therefore the agents can condition their strategies on their opponents' identities for every possible opponent. However, given their limited cognitive abilities, it seems unrealistic to assume that all agents can access all others' identities any time at no cost, especially in a large population. We use the concept of memory to model the limitations of the cognitive ability of the agents. The limitation is the maximum number of agents that an agent can have memory of, which we refer to as the individual memory size ($m$). In other words, the agents can memorize the identities of at most $m$ other agents, but not those of the rest. Once encountered, agents can recognize the identities of those who they remember. Therefore, on the one hand, the agents can prepare a corresponding strategy for each of the agents that they have a memory of, conditioned on the basis of their identities. On the other hand, they can have only one universal strategy of responding to the rest of the agents, whom they do not have a memory of, because they have no way to distinguish these opponents.

With limited memory, the predetermined social rank and the social convention will be effective for the agents only in situations when they are facing opponents with identities that they memorize. This is because social convention provides agents with suggestions on their actions according to their relative ranks to their opponents. If an agent does not know the identity of its opponent, it cannot act accordingly as suggested by the social convention.\footnote{It is still possible for the agents to infer the relative ranks between themselves and their opponents without having the identities of their opponents in memory. For example, the top ranked agent can infer that all its opponents have lower ranks; the bottom ranked agent can infer that all its opponents have higher ranks; the second top ranked agent who has the top ranked agent in its memory can infer that all other opponents have lower ranks. Nevertheless, inferring about the ranks of those agents that an agent has no memory of arguably requires strong cognitive ability, which cannot simply be modeled as memory. Hence, we do not consider such a possibility in this paper.} 


To summarize, in our model, a population of agents is randomly matched in pairs to play the HD game, as shown in Figure 1. Each agent in the population carries a unique number as its identity. The numbers provide a natural linear social rank among the agents. There is a social convention suggesting that agents play $Hawk$ against opponents with lower ranks, and play $Dove$ against opponents with higher ranks. Each agent can memorize the identities of a limited number of agents. We assume that all the agents in the population have the same memory capacity. Note that an agent is not required to use all of its memory. In addition, we do not explicitly model the cost of memorizing an agent's identity. Nevertheless, if we encounter two equilibria that induce the same level of total fitness of the population and one equilibrium requires fewer memory slots than the other, we consider that the former is favored by natural selection. 

\subsection{Equilibrium}

\subsubsection{Memory table, strategy table and social structure graph}

To discuss the potential equilibria of the model, we first need to clarify how we describe the agents' strategies. At the beginning of the game, the agents must decide: (1) who they memorize, and (2) what their corresponding strategies are against each possible opponent. They can have a separate strategy for each opponent that they memorize (because once encountered, they are able to tell which opponent they are playing against if they have the opponent in their memory), but only one universal strategy for all other opponents that they do not have a memory of. As a result, the strategy profiles of the game should contain every agents' choices on memories and the corresponding strategies.

We use two tables, a memory table and a strategy table to represent a particular strategy profile. Figure 2 shows an example of the memory and strategy tables of a particular strategy profile in a population with five agents. In the example, Player 1 ($P_1$) memorizes $P_2$ and $P_3$ and, always plays $Hawk$($H$). $P_2$ and $P_3$ memorize $P_1$ and $P_5$, and play $Dove$($D$) to $P_1$ and the mixed strategy in the mixed NE of the original HD game (\textit{plays $Hawk$ with a probability of V/C}), which we refer to as $M$, to others. $P_4$ memorizes $P_1$ and $P_5$, and plays $D$ to $P_1$, $H$ to $P_5$, and $M$ to others. $P_5$ memorizes $P_1$ and $P_4$, and plays $D$ to $P_1$ and $P_4$, and $M$ to others.

The memory and strategy tables combined can be used to describe any strategy profile in the game. The two tables need to be consistent in describing the same strategy profile of the game, meaning that agents have to be able to perform the strategy they choose in the strategy table, given their memories as shown in the memory table. Specifically, every agent should have the same strategy in the strategy table for all opponents that they do not have a memory of. The strategy table lists every agent's strategies in the HD game, which can be $H$, $D$ or any mix of $H$ and $D$. However, if a strategy table describes a strategy profile that is an equilibrium, then its only possible entries are $H$, $D$ and $M$, and any pair that is on the symmetric positions relative to the diagonal (e.g., (1,2) and (2,1), or (3,5) and (5,3)) must be $(H,D)$, $(D,H)$ or $(M,M)$, reflecting only three NEs in the original HD game.

In addition, if the agents' strategies do not violate the social convention (that is, the upper-right entries in the strategy table must be $H$ if their corresponding entries in the memory table are $1$, and the lower left entries in the strategy table must be $D$ if their corresponding entries in the memory table are $1$), then the strategy profile is an equilibrium that follows social convention. 

\begin{figure}[htbp]
	\centering
	\caption{Example of the memory and strategy tables of a strategy profile}\label{twomatrices}
\begin{minipage}{0.5\textwidth}
\begin{game}{5}{5}[][A: Memory Table]
   	    &  $P_1$                &  $P_2$ & $P_3$ & $P_4$ & $P_5$\\
   	 $P_1$  &    *     & 1 & 1 & 0 & 0\\
   	 $P_2$  &    1     & * & 0 & 0 & 1\\
   	$P_3$  &    1      & 0 & * & 0 & 1\\
   	$P_4$  &    1      & 0 & 0 & * & 1\\
   	    $P_5$  &    1  & 0 & 0 & 1 & *\\
\end{game}
\end{minipage}
\begin{minipage}{0.3\textwidth}

\begin{game}{5}{5}[][B: Strategy Table]
   	    &  $P_1$                &  $P_2$ & $P_3$ & $P_4$ & $P_5$\\
   	 $P_1$  &    *     & H & H & H & H\\
   	 $P_2$  &    D    & * & M & M & M\\
   	$P_3$  &   D      & M & * & M & M\\
   	$P_4$  &    D     & M & M & * & H\\
   	    $P_5$  &   D  & M & M & D & *\\
\end{game}
\end{minipage}
\end{figure}

Since in any equilibrium, the only strategy pairs between any two agents are $(H,D)$, $(D,H)$ or $(M,M)$, there are only two possible relations between any two agents as the result of an equilibrium: one (who plays $H$) dominates the other (who plays $D$) or an equal status (play $M$ against each other). Therefore, a network graph can be used to show the underlying social structure based on the equilibrium of the game. Figure 3 shows the social structure graph based on the example in Figure 2. In the graph, players are represented by circles with names, and a solid arrow between two circles indicates the dominant-submissive relation, where its direction shows the direction of domination. Two players are connected by a dotted line if they have an equal status in which they play $M$ when facing each other. Note that the strategy profile shown in Figure 2 is an equilibrium. However, it does not follow social convention because $P_2$ and $P_3$ both have memory of $P_5$ but do not play \textit{Hawk} against it. The social convention requires that an agent will always play \textit{Hawk} against any agent that they recognize who is lower in the social order.

\begin{figure}[htbp]
	\centering
	\caption{The social structure graph of the society in Figure 2}
\begin{tikzpicture}[
            > = stealth, 
            shorten > = 1pt, 
            auto,
            node distance = 3cm, 
            semithick 
        ]

        \tikzstyle{every state}=[
            draw = black,
            thick,
            fill = white,
            minimum size = 4mm
        ]

        \node[state] (P1) {$P_1$};
        \node[state] (P2) [below left of=P1] {$P_2$};
       \node[state] (P3) [below of=P2] {$P_3$};

       \node[state] (P4) [below right of=P1] {$P_4$};
       \node[state] (P5) [below of=P4] {$P_5$};

        \path[->] (P1) edge node {} (P2);
        \path[->] (P1) edge node {} (P3);
        \path[->] (P1) edge node {} (P4);
        \path[->] (P1) edge node {} (P5); 
        \path[-,dotted] (P2) edge node {} (P3);
        \path[-,dotted] (P2) edge node {} (P4);
        \path[-,dotted] (P2) edge node {} (P5);
        \path[-,dotted] (P3) edge node {} (P4);
        \path[-,dotted] (P3) edge node {} (P5);
        \path[->] (P4) edge node {} (P5);
    \end{tikzpicture}
    \end{figure}
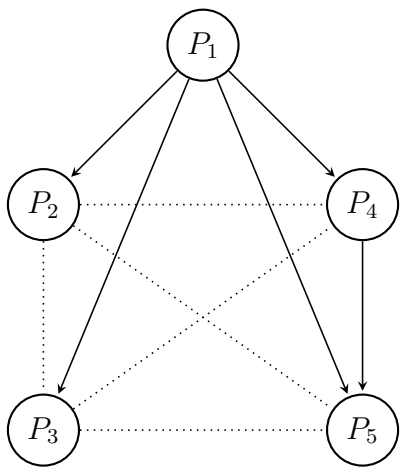
    
\subsection{Equilibrium selection based on total fitness}

Each equilibrium constructs a hierarchical social structure in the sense that it assigns each agent a hierarchical position resulting in some agents enjoying a higher fitness than the others. The assignment is purely based on the agents' identities, which have nothing to do with superior rationality, information, or contribution.

However, the benefit of forming a hierarchical social structure may be more significant at the population level than at the individual level. Indeed, at the individual level, some agents  (those with dominant roles in the hierarchies) may benefit and the rest (those with submissive roles) may suffer from a hierarchical social structure compared to an anarchical state where the only equilibrium is everyone always playing $M$. However, at the population level, a hierarchical social structure may increase the total fitness of the entire population (the sum of individual fitnesses), which helps the population stand out in the competition with other populations, if there are multiple populations. Hence, a population with a hierarchical social structure may be favored by natural selection. 

We seek a social structure that maximizes the total fitness of the population. Recall that $(Hawk, Hawk)$ is the only strategy profile that generates fitness loss, and under our potential equilibrium social structures, this only happens when two players have equal status and they play $(M,M)$ (then $(Hawk,Hawk)$ occurs at a probability of $V^2/C^2$). Therefore, an optimal social structure must have the fewest pairs of agents with the same social status. 

A secondary evaluation of the total fitness of the population is performed to examine the total memory usage among all agents. Memorizing an agent's identity is potentially costly, although the cost may be minimal compared with the fitness loss from the $Hawk-Hawk$ clash in the game. Hence, we only use such an evaluation as a tie-breaker for social structures that provide the same level of total fitness. 

\section{Analysis}

In this section, we study the model with a focus on the hierarchical social structure that maximizes total fitness under various levels of limited cognitive abilities. As described above, the limitation on cognitive abilities is modelled as a memory constraint ($m$), that is, the maximum number of other agents that an agent can remember the identities of.  


\subsection{Fully restricted memory: $m<1$}

First, we consider the case in which memory size is smaller than one. This occurs when agents' cognitive ability is sufficiently low. Since the memory size is not sufficient for the agents to memorize the identity of any single agent, the only possible memory table for this population is as shown in Figure 4(A) (we use a population of five agents for illustration). Then, the only equilibrium strategy table that can be supported in this case is where everyone plays the mixed strategy in the mixed NE in the original HD game ($M$) to all others, as shown in Figure 4 (B). Thus, the corresponding social structure gives an anarchical (or some refer to it as egalitarian) social structure, as shown in Figure 5. 

When evaluating the total fitness, because every pair in the population plays $(M,M)$, the probability of having a $Hawk-Hawk$ clash, the major source of fitness loss in the game, is $V^2/C^2$. Therefore, the expected fitness loss for each pair is $(V^2/C^2)(C) = V^2/C$ and all agents have an identical average fitness value of $(1-V/C)(V/2)$. Note that in this case, the social convention is not followed as the agents cannot access others' identities. 

\begin{figure}[htbp]
	\centering
	\caption{Fully restricted memory}\label{twomatrices}
\begin{minipage}{0.5\textwidth}
\begin{game}{5}{5}[][A: Memory Table]
   	    &  $P_1$                &  $P_2$ & $P_3$ & $P_4$ & $P_5$\\
   	 $P_1$  &    *     & 0 & 0 & 0 & 0\\
   	 $P_2$  &    0     & * & 0 & 0 & 0\\
   	$P_3$  &    0      & 0 & * & 0 & 0\\
   	$P_4$  &    0      & 0 & 0 & * & 0\\
   	    $P_5$  &    0  & 0 & 0 & 0 & *\\
\end{game}
\end{minipage}
\begin{minipage}{0.3\textwidth}

\begin{game}{5}{5}[][B: Strategy Table]
   	    &  $P_1$                &  $P_2$ & $P_3$ & $P_4$ & $P_5$\\
   	 $P_1$  &    *     & M & M & M & M\\
   	 $P_2$  &    M    & * & M & M & M\\
   	$P_3$  &   M      & M & * & M & M\\
   	$P_4$  &    M     & M & M & * & M\\
   	    $P_5$  &   M  & M & M & M & *\\
\end{game}
\end{minipage}
\end{figure}

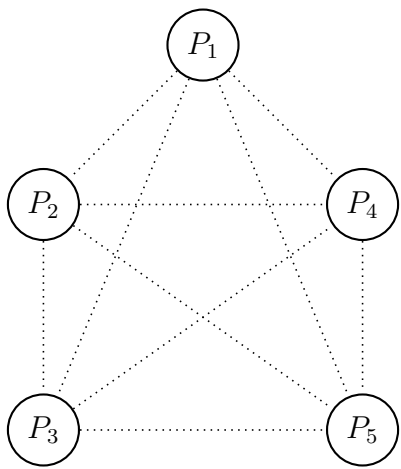
\begin{figure}[h!tbp]
	\centering
	\caption{Social structure, $m <1$}
\begin{tikzpicture}[
            > = stealth, 
            shorten > = 1pt, 
            auto,
            node distance = 3cm, 
            semithick 
        ]

        \tikzstyle{every state}=[
            draw = black,
            thick,
            fill = white,
            minimum size = 4mm
        ]

        \node[state] (P1) {$P_1$};
        \node[state] (P2) [below left of=P1] {$P_2$};
       \node[state] (P3) [below of=P2] {$P_3$};

       \node[state] (P4) [below right of=P1] {$P_4$};
       \node[state] (P5) [below of=P4] {$P_5$};

        \path[-,dotted] (P1) edge node {} (P2);
        \path[-,dotted] (P1) edge node {} (P3);
        \path[-,dotted] (P1) edge node {} (P4);
        \path[-,dotted] (P1) edge node {} (P5); 
        \path[-,dotted] (P2) edge node {} (P3);
        \path[-,dotted] (P2) edge node {} (P4);
        \path[-,dotted] (P2) edge node {} (P5);
        \path[-,dotted] (P3) edge node {} (P4);
        \path[-,dotted] (P3) edge node {} (P5);
        \path[-,dotted] (P4) edge node {} (P5);
    \end{tikzpicture}
    \end{figure}

\subsection{Unlimited Memory: $m \geq N-1$}

The agents are able to memorize all other agents' identities if their memory size equals the population size minus 1 ($N-1$). In this case, with everyone's identity in memory, each agent can condition its strategy for every possible opponent, making $(Hawk,Dove)$ and $(Dove,Hawk)$, in addition to $(M,M)$, possible equilibria in meetings between any two agents. Thus, in the population, there are $2^{\frac{N(N-1)}{2}}$ different ways to form a social structure in which each pair of agents is playing $(H,D)$ or $(D,H)$. Under such a social structure, every pair of agents is arranged with a dominant-submissive relation. Hence, there is no fitness loss in the population from the $Hawk-Hawk$ clash. We refer to these fitness-loss-free social structures as ordered social structures.
    
With the presence of linear rank in their identities and the fact that agents memorize the identities of all other agents, a linear social hierarchical structure  (Figure 7) is the only ordered social structure that follows social convention. Figures 6 (A) and 6(B) show the corresponding memory table and strategy stable. Again, we use a population of five agents for illustration purposes. 

\begin{figure}[htbp]
	\centering
	\caption {Unlimited Memory ($m \geq N-1$), linear hierarchy}\label{twomatrices}
\begin{minipage}{0.5\textwidth}
\begin{game}{5}{5}[][A: Memory Table]
   	    &  $P_1$                &  $P_2$ & $P_3$ & $P_4$ & $P_5$\\
   	 $P_1$  &    *     & 1 & 1 & 1 & 1\\
   	 $P_2$  &    1     & * & 1 & 1 & 1\\
   	$P_3$  &    1      & 1 & * & 1 & 1\\
   	$P_4$  &    1      & 1 & 1 & * & 1\\
   	    $P_5$  &    1  & 1 & 1 & 1 & *\\
\end{game}
\end{minipage}
\begin{minipage}{0.3\textwidth}

\begin{game}{5}{5}[][B: Strategy Table]
   	    &  $P_1$                &  $P_2$ & $P_3$ & $P_4$ & $P_5$\\
   	 $P_1$  &    *     & H & H & H & H\\
   	 $P_2$  &    D    & * & H & H & H\\
   	$P_3$  &   D      & D & * & H & H\\
   	$P_4$  &    D     & D & D & * & H\\
   	    $P_5$  &   D  & D & D & D & *\\
\end{game}
\end{minipage}
\end{figure}

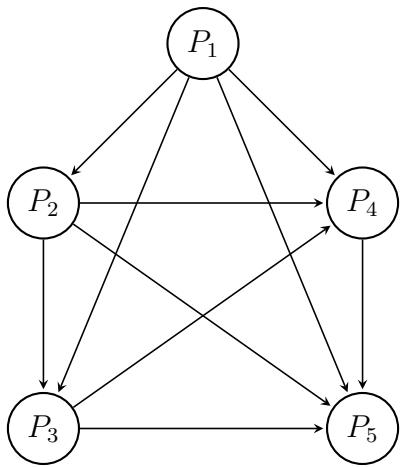
\begin{figure}[h!tbp]
	\centering
	\caption{Linear Hierarchy}
\begin{tikzpicture}[
            > = stealth, 
            shorten > = 1pt, 
            auto,
            node distance = 3cm, 
            semithick 
        ]

        \tikzstyle{every state}=[
            draw = black,
            thick,
            fill = white,
            minimum size = 4mm
        ]

        \node[state] (P1) {$P_1$};
        \node[state] (P2) [below left of=P1] {$P_2$};
       \node[state] (P3) [below of=P2] {$P_3$};

       \node[state] (P4) [below right of=P1] {$P_4$};
       \node[state] (P5) [below of=P4] {$P_5$};

        \path[->]  (P1) edge node {} (P2);
        \path[->]  (P1) edge node {} (P3);
        \path[->]  (P1) edge node {} (P4);
        \path[->]  (P1) edge node {} (P5); 
        \path[->]  (P2) edge node {} (P3);
        \path[->]  (P2) edge node {} (P4);
        \path[->]  (P2) edge node {} (P5);
        \path[->]  (P3) edge node {} (P4);
        \path[->]  (P3) edge node {} (P5);
        \path[->]  (P4) edge node {} (P5);
    \end{tikzpicture}
    \end{figure}

\subsection{Sufficient Memory: $\lceil N/2-1\rceil \leq m < N-1$}

We consider agents as having sufficient memory if they are able to memorize at least half minus one ($\lceil N/2-1\rceil$), but not all the identities of the other agents.\footnote{$\lceil x \rceil$ is the ceiling function, which gives the least integer greater than or equal to $x$.} We demonstrate that sufficient memory is sufficient for a population to form an ordered social structure. In particular, we show that a linear hierarchy can be formed.

\begin{proposition}
A linear hierarchy social structure can be an equilibrium if the memory size ($m$) satisfies $m \geq \lceil N/2-1\rceil$.
\end{proposition}

\begin{proof}
Under a linear hierarchical social structure, each agent plays either $Hawk$ ($H$) or $Dove$ ($D$) to its opponents. To successfully implement these strategies, the agents can choose to memorize the fewer between all agents that they need to play $H$ to and all agents that they need to play $D$ to. Then, they play $H$ (or $D$) to the agents in their memory and $D$ (or $H$) to those who are not in their memory. The required memory level for the agents depends on their positions in the linear social rank. When $N$ is odd, the agent who needs the largest memory size is the one who is positioned in the middle, and it needs to memorize $(N-1)/2$ other agents to play $H$ (or $D$) to, and $D$ (or $H$) to the rest. If $N$ is even, the two agents in the middle need to memorize $N/2-1$ other agents. Hence, if everyone has a memory size no less than $\lceil N/2-1\rceil$, then the linear hierarchical social structure can be supported as an equilibrium.
\end{proof}



Note that if we treat the use of memory as a minor source of fitness loss, the optimal memory size that gives rise to the linear structure is $\lceil N/2-1\rceil$. Moreover, only the agent(s) situated in the middle of the linear social rank require(s) the use of its entire memory, while others can use less. Figure 8 (A) illustrates the optimal memory table for a population of five agents. We call this form of memory usage as a ``triangular memory structure'' because the usage of memory gradually increases as we move from either the top or the bottom toward the middle of the rank of agents, and the agents who need the largest memory size are those who are situated in the middle. The total memory usage is $\lceil N(N-2)/4 \rceil$ for a triangular memory structure in a population of $N$ agents. When compared with Figure 6 (A), one can observe that the memory cost is greatly reduced in Figure 8 (A), and it is sufficient to ensure that the strategy table in Figure 8 (B) (identical to Figure 6 (B)) constitutes an equilibrium that follows social convention. 

\begin{figure}[htbp]
	\centering
	\caption {Sufficient Memory, $\lceil N/2-1\rceil \leq m < N-1$}\label{twomatrices}
\begin{minipage}{0.5\textwidth}
\begin{game}{5}{5}[][A: Memory Table]
   	    &  $P_1$                &  $P_2$ & $P_3$ & $P_4$ & $P_5$\\
   	 $P_1$  &    *     & 0 & 0 & 0 & 0\\
   	 $P_2$  &    1     & * & 0 & 0 & 0\\
   	$P_3$  &    1      & 1 & * & 0 & 0\\
   	$P_4$  &    0      & 0 & 0 & * & 1\\
   	    $P_5$  &    0  & 0 & 0 & 0 & *\\
\end{game}
\end{minipage}
\begin{minipage}{0.3\textwidth}

\begin{game}{5}{5}[][B: Strategy Table]
   	    &  $P_1$                &  $P_2$ & $P_3$ & $P_4$ & $P_5$\\
   	 $P_1$  &    *     & H & H & H & H\\
   	 $P_2$  &    D    & * & H & H & H\\
   	$P_3$  &   D      & D & * & H & H\\
   	$P_4$  &    D     & D & D & * & H\\
   	    $P_5$  &   D  & D & D & D & *\\
\end{game}
\end{minipage}
\end{figure}

\subsection{Insufficient Memory: $1 \leq m < \lceil N/2-1\rceil $}

When the memory size is smaller than $\lceil N/2-1\rceil$, it is impossible for the population to form a linear hierarchical structure that follows social convention. Any equilibrium formed with agents' memory size smaller than $\lceil N/2-1\rceil$ will involve some pairs of agents playing $M$ to each other, causing a fitness loss from the $Hawk-Hawk$ clash. Therefore, the optimal structures are those that induce the fewest pairs of agents playing $(M,M)$ as their equilibrium strategies.

We first examine the case in which the agents can memorize at the most one other agent's identity, which we refer as the singular memory.
We use a computational method\footnote{See the appendix for a description of the computation method. The code can be found at https://github.com/harrisonhhy/optimal\_social\_structure.} to find all equilibrium social structures by examining all possible strategy tables that can be supported by at least one memory table (i.e., the agents must be able to perform the strategy they choose in the strategy table given the memory table). The results show that the despotic social system, where one agent dominates all others ($Hawk-Dove$ relation), is an equilibrium social structure that follows social convention and maximizes the total fitness in a population of five agents (Figures 9 and 10) and in a population of six agents as well.\footnote{Symmetrically, a social system in which one agent is dominated by all others is also an equilibrium social structure that follows social convention and maximizes the total fitness.} Note that in a population with five agents, there exists another equilibrium social structure, which we refer as the ``proxy despotism," where the top-ranked (in the linear social rank) agent only dominates (plays $H$ to) the second top-ranked agent and plays $M$ with the rest. The second top-ranked agent dominates (plays $H$ to) all agents besides the top-ranked agent but plays $D$ to the top.\footnote{Symmetrically, a social system in which the bottom ranked agent is dominated by the second bottom ranked and the second top ranked agent is dominated by all agents beside the bottomed ranked agent is also an equilibrium social structure that follows social convention and maximizes the total fitness.} It is equally as good as the despotic structure in terms of total fitness. As shown in Figure 11, proxy despotism has the same probability of having $Hawk-Hawk$ clash as the despotic social structure (4 out of 10 pairs do not have clash). However, their difference is that the proxy despotism requires full use of all agents' singular memory size, while under the despotic structure, the use of memory for the top-ranked agent can be waived. Consequently, the despotic social structure has a fitness advantage over the proxy despotic structure. This may explain their relative frequencies (\cite{Sasaki2016}). 

The intuition behind the optimality of the despotic social structure is that when the top-ranked agent is memorized by all others, each usage of memory creates a non-clash ($Hawk-Dove$) pairwise interaction between the top-ranked agent and the agent who uses the memory (in the example of a population with five agents, there are four memories used and four corresponding non-clash relations).



More generally, we conjecture that with insufficient memory  ($1 \leq m < \lceil N/2-1\rceil$), the optimal way for agents to use their memories in terms of social efficiency is described as follows.\footnote{We thank an anonymous reviewer for the suggestion.} There are $m$ out of $N$ agents, marked as group A, and other $N-m$ agents, marked as group B. Agents in group A form a clash-free linear hierarchy among themselves, and the most fitness-efficient way to do so is to use the triangular memory structure according to Proposition 1. Under the triangular memory structure, agents in group A will be separated into two halves, where the higher-ranked half (lower-ranked half) agents only memorize agents who have higher (lower) ranks than them, and then they play $D$ ($H$) to the agents in their memory and $H$ ($D$) to those who are not in their memory (including those agents in group B). The $N-m$ agents in group B will memorize (and only memorize) all agents in group A, play $D$ to the higher-ranked half in group A, play $H$ to the lower-ranked half in group A, and play $M$ to other agents in group B, whom they do not memorize, forming an anarchical/egalitarian sub-social structure within group B. With this form of memory structure and the corresponding strategies, all $m(N-m)$ pairs between group A and group B agents and $m(m-1)/2$ pairs among group A agents have $Hawk-Dove$ as their strategies in equilibrium. Hence, there are $(m(N-m)+m(m-1)/2)$ out of $N(N-1)/2$ clash-free pairs, and every one pairwise clash-free relation is maintained by at most one memory. Specifically, the $H-D$ relationship between a group A agent and a group B agent is maintained by the group B agent's memory of the group A agent. The $H-D$ relation between a pair of agents within the higher-ranked (lower-ranked) half in group A is maintained by the lower-ranked (higher-ranked) agent's memory of the higher-ranked (lower-ranked) agent. The $H-D$ relation between one agent from the higher-ranked half and one agent from the lower-ranked half in group A requires no memory to sustain. We call this social structure the ``three-layer dominance hierarchy'' because society is separated into three classes. When $m$ is even, the upper class contains the top $m/2$ ranked agents with a linear hierarchy, the middle class contains the middle $N-m$ ranked agents with an egalitarian social structure, and the lower class contains the bottom $m/2$ ranked agents with a linear hierarchy. When $m$ is odd, the upper class contains the top $\lceil m/2 \rceil$ (or $\lceil m/2 \rceil-1$) ranked agents, and the lower class contains the bottom $\lceil m/2 \rceil-1$ (or $\lceil m/2 \rceil$) ranked agents. The total amount of memory required for the three-layer dominance hierarchy is $m(N-m) + \lceil m(m-2)/4 \rceil$. Note that the despotic social system is consistent with the description of the three-layer dominance hierarchy for the special case of $m=1$. Figure 13 provides a general illustration of the ``three-layer dominance hierarchy"


We confirm our conjecture in a population of seven agents with $m=2$, using our computation method. The social structure at equilibrium that induces the least fitness loss is that the top-ranked agent dominates all others, the bottom-ranked agent is dominated by all others, and the remaining five agents form an egalitarian social structure among themselves. This is supported by a memory structure where the five middle-ranked agents memorize the top-ranked and the bottom-ranked agents, and the top- and bottom-ranked agents do not memorize any agent (Figures 14 and 15).

\begin{figure}[htbp]
	\centering
	\caption {Singular Memory, Despotic Structure $m=1$}\label{twomatrices}
\begin{minipage}{0.5\textwidth}
\begin{game}{5}{5}[][A: Memory Table]
   	    &  $P_1$                &  $P_2$ & $P_3$ & $P_4$ & $P_5$\\
   	 $P_1$  &    *     & 0 & 0 & 0 & 0\\
   	 $P_2$  &    1     & * & 0 & 0 & 0\\
   	$P_3$  &    1      & 0 & * & 0 & 0\\
   	$P_4$  &    1      & 0 & 0 & * & 0\\
   	    $P_5$  &    1  & 0 & 0 & 0 & *\\
\end{game}
\end{minipage}
\begin{minipage}{0.3\textwidth}

\begin{game}{5}{5}[][B: Strategy Table]
   	    &  $P_1$                &  $P_2$ & $P_3$ & $P_4$ & $P_5$\\
   	 $P_1$  &    *     & H & H & H & H\\
   	 $P_2$  &    D    & * & M & M & M\\
   	$P_3$  &   D      & M & * & M & M\\
   	$P_4$  &    D     & M & M & * & M\\
   	    $P_5$  &   D  & M & M & M & *\\
\end{game}
\end{minipage}
\end{figure}


\begin{figure}[h!tbp]
	\centering
	\caption{Despotic social structure, $m = 1$}
\begin{tikzpicture}[
            > = stealth, 
            shorten > = 1pt, 
            auto,
            node distance = 3cm, 
            semithick 
        ]

        \tikzstyle{every state}=[
            draw = black,
            thick,
            fill = white,
            minimum size = 4mm
        ]

        \node[state] (P1) {$P_1$};
        \node[state] (P2) [below left of=P1] {$P_2$};
       \node[state] (P3) [below of=P2] {$P_3$};

       \node[state] (P4) [below right of=P1] {$P_4$};
       \node[state] (P5) [below of=P4] {$P_5$};

        \path[->] (P1) edge node {} (P2);
        \path[->] (P1) edge node {} (P3);
        \path[->] (P1) edge node {} (P4);
        \path[->] (P1) edge node {} (P5); 
        \path[-,dotted] (P2) edge node {} (P3);
        \path[-,dotted] (P2) edge node {} (P4);
        \path[-,dotted] (P2) edge node {} (P5);
        \path[-,dotted] (P3) edge node {} (P4);
        \path[-,dotted] (P3) edge node {} (P5);
        \path[-,dotted] (P4) edge node {} (P5);
    \end{tikzpicture}
    \end{figure}
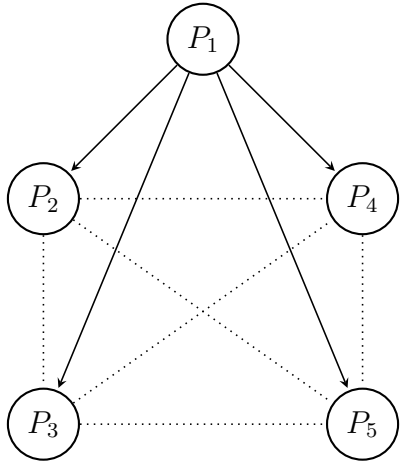
    
    
\begin{figure}[htbp]
	\centering
	\caption {Singular Memory, Proxy Despotism $m=1$}\label{twomatrices}
\begin{minipage}{0.5\textwidth}
\begin{game}{5}{5}[][A: Memory Table]
   	    &  $P_1$                &  $P_2$ & $P_3$ & $P_4$ & $P_5$\\
   	 $P_1$  &    *     & 1 & 0 & 0 & 0\\
   	 $P_2$  &    1     & * & 0 & 0 & 0\\
   	$P_3$  &    0      & 1 & * & 0 & 0\\
   	$P_4$  &    0      & 1 & 0 & * & 0\\
   	    $P_5$  &    0  & 1 & 0 & 0 & *\\
\end{game}
\end{minipage}
\begin{minipage}{0.3\textwidth}

\begin{game}{5}{5}[][B: Strategy Table]
   	    &  $P_1$                &  $P_2$ & $P_3$ & $P_4$ & $P_5$\\
   	 $P_1$  &    *     & H & M & M & M\\
   	 $P_2$  &    D    & * & H & H & H\\
   	$P_3$  &   M      & D & * & M & M\\
   	$P_4$  &    M     & D & M & * & M\\
   	    $P_5$  &   M  & D & M & M & *\\
\end{game}
\end{minipage}
\end{figure}


\begin{figure}[h!tbp]
	\centering
	\caption{Proxy Despotism $m = 1$}
\begin{tikzpicture}[
            > = stealth, 
            shorten > = 1pt, 
            auto,
            node distance = 3cm, 
            semithick 
        ]

        \tikzstyle{every state}=[
            draw = black,
            thick,
            fill = white,
            minimum size = 4mm
        ]

        \node[state] (P1) {$P_1$};
        \node[state] (P2) [below left of=P1] {$P_2$};
       \node[state] (P3) [below of=P2] {$P_3$};

       \node[state] (P4) [below right of=P1] {$P_4$};
       \node[state] (P5) [below of=P4] {$P_5$};

        \path[->] (P1) edge node {} (P2);
        \path[-,dotted] (P1) edge node {} (P3);
        \path[-,dotted] (P1) edge node {} (P4);
        \path[-,dotted] (P1) edge node {} (P5); 
        \path[->] (P2) edge node {} (P3);
        \path[->] (P2) edge node {} (P4);
        \path[->] (P2) edge node {} (P5);
        \path[-,dotted] (P3) edge node {} (P4);
        \path[-,dotted] (P3) edge node {} (P5);
        \path[-,dotted] (P4) edge node {} (P5);
    \end{tikzpicture}
    \end{figure}
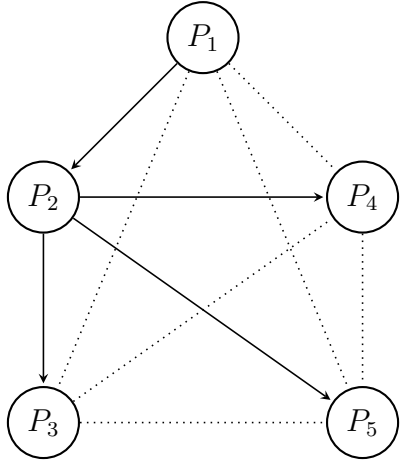
    
    \begin{figure} [htbp]
\centering
\caption {Three-layer dominance hierarchy: a general illustration}
\includegraphics[scale=0.4]{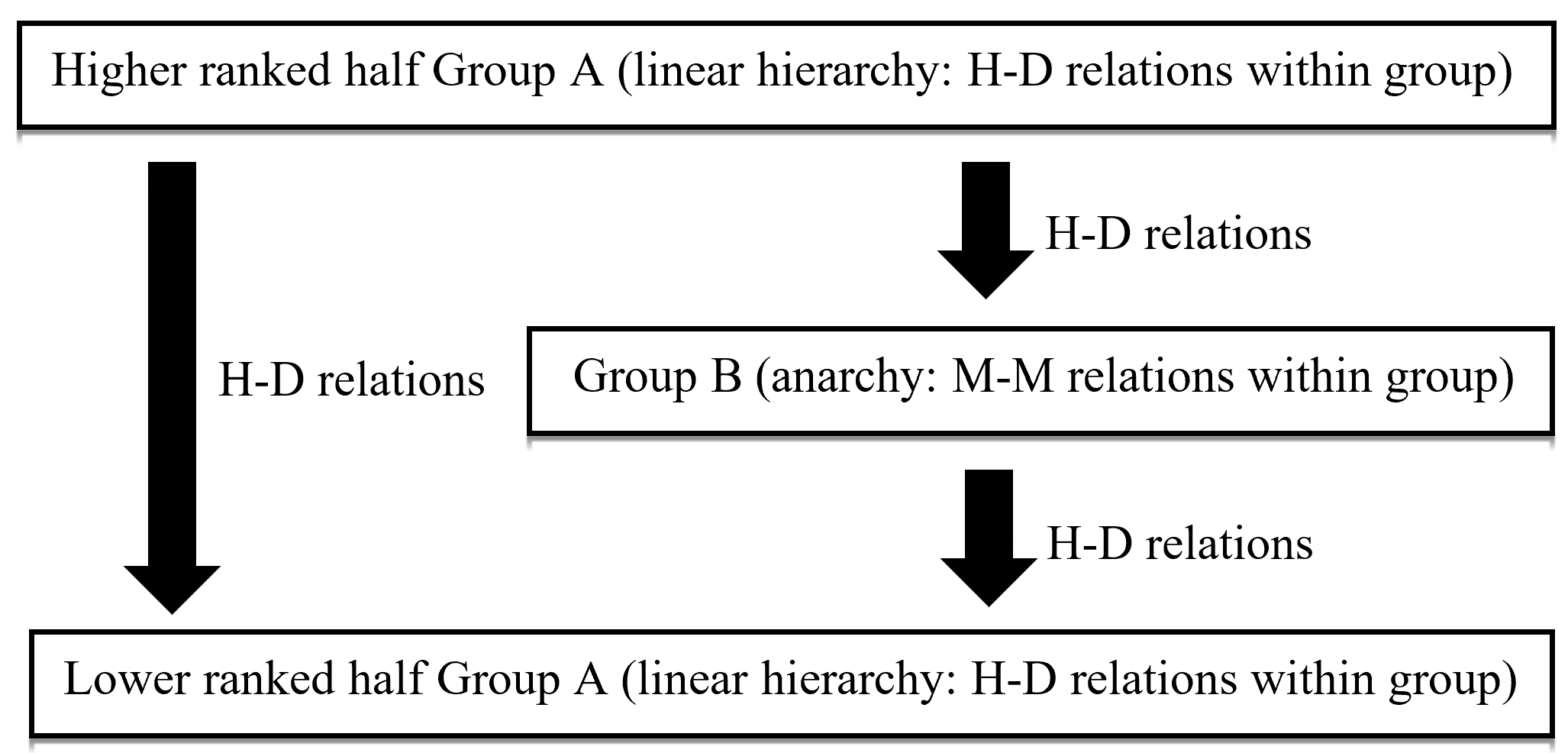}
\end{figure}
    
\begin{figure}[htbp]
	\centering
	\caption {Insufficient Memory, $N=7, m=2$}\label{twomatrices}
\begin{minipage}{0.5\textwidth}
\begin{game}{7}{7}[][A: Memory Table]
   	    &  $P_1$                &  $P_2$ & $P_3$ & $P_4$ & $P_5$ & $P_6$ & $P_7$\\
   	 $P_1$  &    *     & 0 & 0 & 0 & 0 & 0 & 0\\
   	 $P_2$  &    1     & * & 0 & 0 & 0 & 0 & 1\\
   	$P_3$  &    1      & 0 & * & 0 & 0 & 0 & 1\\
   	$P_4$  &    1      & 0 & 0 & * & 0 & 0 & 1\\
   	$P_5$  &    1  & 0 & 0 & 0 & * & 0 & 1\\
    $P_6$  &    1  & 0 & 0 & 0 & 0 & * & 1\\
   	$P_7$  &    0  & 0 & 0 & 0 & 0 & 0 & *\\
\end{game}
\end{minipage}
\begin{minipage}{0.45\textwidth}

\begin{game}{7}{7}[][N: Strategy Table]
   	    &  $P_1$                &  $P_2$ & $P_3$ & $P_4$ & $P_5$ & $P_6$ & $P_7$\\
   	 $P_1$  &    *     & H & H & H & H & H & H\\
   	 $P_2$  &    D     & * & M & M & M & M & H\\
   	$P_3$  &    D      & M & * & M & M & M & H\\
   	$P_4$  &    D      & M & M & * & M & M & H\\
   	$P_5$  &    D  & M & M & M & * & M & H\\
    $P_6$  &    D  & M & M & M & M & * & H\\
   	$P_7$  &    D  & D & D & D & D & D & *\\
\end{game}
\end{minipage}
\end{figure}

\begin{figure}[h!tbp]
	\centering
	\caption{Three-layer dominance hierarchy; $m = 2$}
\begin{tikzpicture}[
            > = stealth, 
            shorten > = 1pt, 
            auto,
            node distance = 3cm, 
            semithick 
        ]

        \tikzstyle{every state}=[
            draw = black,
            thick,
            fill = white,
            minimum size = 4mm
        ]

        \node[state] (P2) {$P_2$};
        \node[state] (P1) [below left of=P2] {$P_1$};
       \node[state] (P4) [below right of=P1] {$P_4$};
                  \node[state] (P5) [right of=P4] {$P_5$};

              \node[state] (P6) [right of=P5] {$P_6$};
                     \node[state] (P7) [above right of=P6] {$P_7$};
                                          \node[state] (P3) [above left of=P7] {$P_3$};

        \path[->] (P1) edge node {} (P2);
                \path[->] (P1) edge node {} (P3);
                        \path[->] (P1) edge node {} (P4);
                                \path[->] (P1) edge node {} (P5);
                                        \path[->] (P1) edge node {} (P6);
                                                \path[->] (P1) edge node {} (P7);
                                                    
        \path[->] (P2) edge node {} (P7);
                        \path[->] (P3) edge node {} (P7);
                                \path[->] (P4) edge node {} (P7);
                                        \path[->] (P5) edge node {} (P7);
                                                \path[->] (P6) edge node {} (P7);
        \path[-,dotted] (P2) edge node {} (P3);
        \path[-,dotted] (P2) edge node {} (P4);
                \path[-,dotted] (P2) edge node {} (P5);
        \path[-,dotted] (P2) edge node {} (P6);
        \path[-,dotted] (P3) edge node {} (P4);
        \path[-,dotted] (P3) edge node {} (P5);
                \path[-,dotted] (P3) edge node {} (P6);
        \path[-,dotted] (P4) edge node {} (P5);
        \path[-,dotted] (P4) edge node {} (P6);
               \path[-,dotted] (P5) edge node {} (P6);
    \end{tikzpicture}
    \end{figure}
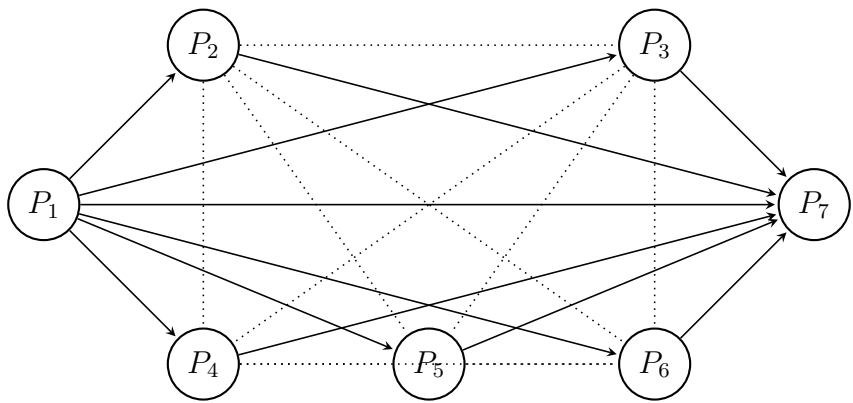



\newpage

\section{Extension} \label{sec:extension}

Thus far, we have analyzed the optimal social hierarchical structures under various levels of cognitive ability limitation, where there exists a predetermined linear social rank and an associated social convention with regard to the identities of the agents. One might wonder what can change in those equilibria if the rank or the social convention does not exist -- in which case the agents' identities are identical \textit{ex ante}, although they are still unique tags that others can recognize and condition their strategies on. 

Treating identities as arbitrary name tags instead of numbers with rank enables agents to form social hierarchical structures that contain cyclic dominance relations (loop) where A dominates B, B dominates C, and C dominates A. This type of intransitive dominance relation is ruled out if everyone follows social convention, which is based on transitive linear rank. A cyclic loop can sometimes help a population achieve a higher total fitness level with fewer memory usages, especially when the population size is small. Figure 16 shows an example of a population of five agents forming a social structure containing a loop (among $P_3$, $P_4$, and $P_5$). Its individual memory usage is at the same level as that of the despotic social structure (Figure 9), but it achieves a clash-free status as the linear hierarchy does. See Figure 17. Apparently, this structure with a loop outraces all structures introduced in the previous section.

Various studies (\cite{Banks1956}; \cite{Chase1982}; \cite{Wang2011}) have pointed out that non-transitive dominance relations are very rare in nature compared to transitive linear dominance structures. This suggests that there are deeper reasons for species not to form non-transitive dominance relations. Besides directly ruling out a structure with loops by assuming that agents just follow the predetermined social convention (as in equilibrium, they do not have reasons not to follow), we do not yet have a feature in our model to explain why non-transitive structures are not favored by agents despite their potential to achieve higher total fitness.  

\begin{figure}[htbp]
	\centering
	\caption {Singular Memory, Loop, $m=1$}\label{twomatrices}
\begin{minipage}{0.5\textwidth}
\begin{game}{5}{5}[][A: Memory Table]
   	    &  $P_1$                &  $P_2$ & $P_3$ & $P_4$ & $P_5$\\
   	 $P_1$  &    *     & 0 & 0 & 0 & 0\\
   	 $P_2$  &    1     & * & 0 & 0 & 0\\
   	$P_3$  &    0      & 0 & * & 1 & 0\\
   	$P_4$  &    0      & 0 & 0 & * & 1\\
   	    $P_5$  &    0  & 0 & 1 & 0 & *\\
\end{game}
\end{minipage}
\begin{minipage}{0.3\textwidth}

\begin{game}{5}{5}[][B: Strategy Table]
   	    &  $P_1$                &  $P_2$ & $P_3$ & $P_4$ & $P_5$\\
   	 $P_1$  &    *     & H & H & H & H\\
   	 $P_2$  &    D    & * & H & H & H\\
   	$P_3$  &   D      & D & * & H & D\\
   	$P_4$  &    D     & D & D & * & H\\
   	    $P_5$  &   D  & D & H & D & *\\
\end{game}
\end{minipage}
\end{figure}

\begin{figure}[h!tbp]
	\caption{Loop Structure}
	\centering
\begin{tikzpicture}[
            > = stealth, 
            shorten > = 1pt, 
            auto,
            node distance = 3cm, 
            semithick 
        ]

        \tikzstyle{every state}=[
            draw = black,
            thick,
            fill = white,
            minimum size = 4mm
        ]

        \node[state] (P1) {$P_1$};
        \node[state] (P2) [below left of=P1] {$P_2$};
       \node[state] (P3) [below of=P2] {$P_3$};

       \node[state] (P4) [below right of=P1] {$P_4$};
       \node[state] (P5) [below of=P4] {$P_5$};

        \path[->] (P1) edge node {} (P2);
        \path[->] (P1) edge node {} (P3);
        \path[->] (P1) edge node {} (P4);
        \path[->] (P1) edge node {} (P5); 
        \path[->] (P2) edge node {} (P3);
        \path[->] (P2) edge node {} (P4);
        \path[->] (P2) edge node {} (P5);
        \path[->] (P3) edge node {} (P4);
        \path[->] (P5) edge node {} (P3);
        \path[->] (P4) edge node {} (P5);
    \end{tikzpicture}
    \end{figure}
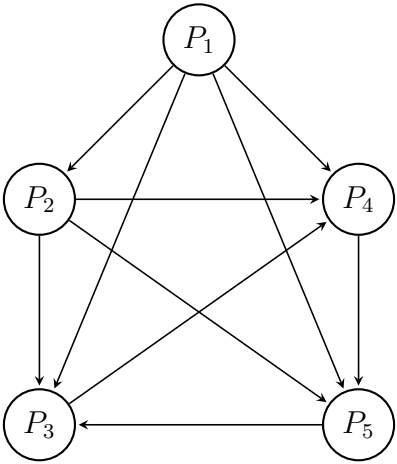

\section{Conclusion} \label{sec:conclusion}

We propose a model in which a population of agents is matched to play a Hawk-Dove game. Agents are equipped with unique identities, and there exists a linear social rank over their identities, which is accompanied by a predetermined social convention. Our model suggests that at different levels of cognitive ability limitations, different hierarchical social structures can be supported as equilibria that follow social convention, and they are optimal in terms of the total fitness of the population. Our findings can be supported by the fact that these hierarchies are the most common ones observed in nature. 
Our model suggests a way to understand how different species utilize their cognitive abilities in social interactions by examining the existing hierarchical social structures in their populations.    

\singlespacing
\bibliography{references}

\begin{thebibliography}{31}
\providecommand{\natexlab}[1]{#1}
\providecommand{\url}[1]{\texttt{#1}}
\expandafter\ifx\csname urlstyle\endcsname\relax
  \providecommand{\doi}[1]{doi: #1}\else
  \providecommand{\doi}{doi: \begingroup \urlstyle{rm}\Url}\fi

\bibitem[Addison and Simmel(1980)]{Addison1980}
W.~E. Addison and E.~C. Simmel.
\newblock {The relationship between dominance and leadership in a flock of
  ewes}.
\newblock \emph{Bulletin of the Psychonomic Society}, 1980.
\newblock ISSN 00905054.
\newblock \doi{10.3758/BF03334540}.

\bibitem[Alcock(2013)]{Alcock2013}
J.~Alcock.
\newblock \emph{{Animal behavior: An evolutionary approach}}.
\newblock 2013.

\bibitem[Appleby(1983)]{Appleby1983}
M.~C. Appleby.
\newblock {The probability of linearity in hierarchies}.
\newblock \emph{Animal Behaviour}, 1983.
\newblock ISSN 00033472.
\newblock \doi{10.1016/S0003-3472(83)80084-0}.

\bibitem[Banks(1956)]{Banks1956}
E.~M. Banks.
\newblock {Social Organization in Red Jungle Fowl Hens (Gallus Gallus Subsp.)}.
\newblock \emph{Ecology}, 1956.
\newblock ISSN 0012-9658.
\newblock \doi{10.2307/1933136}.

\bibitem[Barkan et~al.(1986)Barkan, Craig, Strahl, Stewart, and
  Brown]{Barkan1986}
C.~P. Barkan, J.~L. Craig, S.~D. Strahl, A.~M. Stewart, and J.~L. Brown.
\newblock {Social dominance in communal Mexican jays Aphelocoma ultramarina}.
\newblock \emph{Animal Behaviour}, 1986.
\newblock ISSN 00033472.
\newblock \doi{10.1016/0003-3472(86)90021-7}.

\bibitem[Binmore(1994)]{Binmore1994}
K.~G. Binmore.
\newblock {Game Theory and the Social Contract: Playing Fair}.
\newblock \emph{English}, 1994.

\bibitem[Chase(1982)]{Chase1982}
I.~D. Chase.
\newblock {Dynamics of Hierarchy Formation: The Sequential Development of
  Dominance Relationships}.
\newblock \emph{Behaviour}, 1982.
\newblock ISSN 1568539X.
\newblock \doi{10.1163/156853982X00364}.

\bibitem[Chase and Seitz(2011)]{Chase2011}
I.~D. Chase and K.~Seitz.
\newblock {Self-structuring properties of dominance hierarchies. A new
  perspective}.
\newblock In \emph{Advances in Genetics}. 2011.
\newblock \doi{10.1016/B978-0-12-380858-5.00001-0}.

\bibitem[Chase et~al.(2002)Chase, Tovey, Spangler-Martin, and
  Manfredonia]{Chase2002}
I.~D. Chase, C.~Tovey, D.~Spangler-Martin, and M.~Manfredonia.
\newblock {Individual differences versus social dynamics in the formation of
  animal dominance hierarchies}.
\newblock \emph{Proceedings of the National Academy of Sciences of the United
  States of America}, 2002.
\newblock ISSN 00278424.
\newblock \doi{10.1073/pnas.082104199}.

\bibitem[de~Vries(1995)]{DeVries1995}
H.~de~Vries.
\newblock {An improved test of linearity in dominance hierarchies containing
  unknown or tied relationships}.
\newblock \emph{Animal Behaviour}, 1995.
\newblock ISSN 00033472.
\newblock \doi{10.1016/0003-3472(95)80053-0}.

\bibitem[Doi and Nakamaru(2018)]{DoiNakamaru2018}
K.~Doi and M.~Nakamaru.
\newblock The coevolution of transitive inference and memory capacity in the
  hawk–dove game.
\newblock \emph{Journal of Theoretical Biology}, pages 91--107, 2018.

\bibitem[Drews(1993)]{Drews1993}
C.~Drews.
\newblock {The concept and definition of dominance in animal behaviour}.
\newblock \emph{Behaviour}, 1993.
\newblock ISSN 00057959.
\newblock \doi{10.1163/156853993X00290}.

\bibitem[Dugatkin and Earley(2004)]{Dugatkin2004}
L.~A. Dugatkin and R.~L. Earley.
\newblock {Individual recognition, dominance hierarchies and winner and loser
  effects}.
\newblock \emph{Proceedings of the Royal Society B: Biological Sciences}, 2004.
\newblock ISSN 14712970.
\newblock \doi{10.1098/rspb.2004.2777}.

\bibitem[Favati et~al.(2017)Favati, L{\o}vlie, and Leimar]{Favati2017a}
A.~Favati, H.~L{\o}vlie, and O.~Leimar.
\newblock {Individual aggression, but not winner-loser effects, predicts social
  rank in male domestic fowl}.
\newblock \emph{Behavioral Ecology}, 2017.
\newblock ISSN 14657279.
\newblock \doi{10.1093/beheco/arx053}.

\bibitem[Goessmann et~al.(2000)Goessmann, Hemelrijk, and Huber]{Goessmann2000}
C.~Goessmann, C.~Hemelrijk, and R.~Huber.
\newblock {The formation and maintenance of crayfish hierarchies: Behavioral
  and self-structuring properties}.
\newblock \emph{Behavioral Ecology and Sociobiology}, 2000.
\newblock ISSN 03405443.
\newblock \doi{10.1007/s002650000222}.

\bibitem[Halpern and Pass(2015)]{Halpern2015}
J.~Y. Halpern and R.~Pass.
\newblock {Algorithmic rationality: Game theory with costly computation}.
\newblock \emph{Journal of Economic Theory}, 2015.
\newblock ISSN 10957235.
\newblock \doi{10.1016/j.jet.2014.04.007}.

\bibitem[Hausfater et~al.(1982)Hausfater, Altmann, and Altmann]{Hausfater1982}
G.~Hausfater, J.~Altmann, and S.~Altmann.
\newblock {Long-term consistency of dominance relations among female baboons
  (Papio cynocephalus)}.
\newblock \emph{Science}, 1982.
\newblock ISSN 00368075.
\newblock \doi{10.1126/science.217.4561.752}.

\bibitem[Heinze(1990)]{Heinze1990}
J.~Heinze.
\newblock {Dominance behavior among ant females}.
\newblock \emph{Naturwissenschaften}, 1990.
\newblock ISSN 14321904.
\newblock \doi{10.1007/BF01131799}.

\bibitem[Holekamp and Smale(1993)]{Holekamp1993}
K.~E. Holekamp and L.~Smale.
\newblock {Ontogeny of dominance in free-living spotted hyaenas: Juvenile rank
  relations with other immature individuals}.
\newblock \emph{Animal Behaviour}, 1993.
\newblock ISSN 00033472.
\newblock \doi{10.1006/anbe.1993.1214}.

\bibitem[Kummer(1984)]{Kummer1984}
H.~Kummer.
\newblock {From laboratory to desert and back: A social system of hamadryas
  baboons}.
\newblock \emph{Animal Behaviour}, 1984.
\newblock ISSN 00033472.
\newblock \doi{10.1016/S0003-3472(84)80208-0}.

\bibitem[Kura et~al.(2016)Kura, Broom, and Kandler]{Kura2016}
K.~Kura, M.~Broom, and A.~Kandler.
\newblock {A Game-Theoretical Winner and Loser Model of Dominance Hierarchy
  Formation}.
\newblock \emph{Bulletin of Mathematical Biology}, 2016.
\newblock ISSN 15229602.
\newblock \doi{10.1007/s11538-016-0186-9}.

\bibitem[Nakamaru and Sasaki(2003)]{NakamaruSasaki2003}
M.~Nakamaru and A.~Sasaki.
\newblock Can transitive inference evolve in animals playing the hawk–dove
  game?
\newblock \emph{Journal of Theoretical Biology}, pages 461--470, 2003.

\bibitem[Nelissen(1985)]{Nelissen1985}
M.~H. Nelissen.
\newblock {Structure Of The Dominance Hierarchy and Dominance Determining
  “Group Factors” in Melanochromis Auratus (Pisces, Cichlidae)}.
\newblock \emph{Behaviour}, 1985.
\newblock ISSN 1568539X.
\newblock \doi{10.1163/156853985X00280}.

\bibitem[Sasaki et~al.(2016)Sasaki, Penick, Shaffer, Haight, Pratt, and
  Liebig]{Sasaki2016}
T.~Sasaki, C.~A. Penick, Z.~Shaffer, K.~L. Haight, S.~C. Pratt, and J.~Liebig.
\newblock {A simple behavioral model predicts the emergence of complex animal
  hierarchies}.
\newblock \emph{American Naturalist}, 2016.
\newblock ISSN 00030147.
\newblock \doi{10.1086/686259}.

\bibitem[Savin-Williams(1980)]{Savin-Williams1980}
R.~C. Savin-Williams.
\newblock {Dominance hierarchies in groups of middle to late adolescent males}.
\newblock \emph{Journal of Youth and Adolescence}, 1980.
\newblock ISSN 00472891.
\newblock \doi{10.1007/BF02088381}.

\bibitem[Schjelderup-Ebbe(1935)]{Schjelderup-Ebbe1935}
T.~Schjelderup-Ebbe.
\newblock {Social behavior of birds.}
\newblock In \emph{A Handbook of Social Psychology}. 1935.

\bibitem[Smith and Parker(1976)]{Smith1976}
J.~M. Smith and G.~A. Parker.
\newblock {The logic of asymmetric contests}.
\newblock \emph{Animal Behaviour}, 1976.
\newblock ISSN 00033472.
\newblock \doi{10.1016/S0003-3472(76)80110-8}.

\bibitem[Smith and Price(1973)]{Smith1973}
J.~M. Smith and G.~R. Price.
\newblock {The logic of animal conflict}.
\newblock \emph{Nature}, 1973.
\newblock ISSN 00280836.
\newblock \doi{10.1038/246015a0}.

\bibitem[Surbeck et~al.(2011)Surbeck, Mundry, and Hohmann]{Surbeck2011}
M.~Surbeck, R.~Mundry, and G.~Hohmann.
\newblock {Mothers matter! Maternal support, dominance status and mating
  success in male bonobos (Pan paniscus)}.
\newblock In \emph{Proceedings of the Royal Society B: Biological Sciences},
  2011.
\newblock \doi{10.1098/rspb.2010.1572}.

\bibitem[Vannini and Sardini(1971)]{Vannini1971}
M.~Vannini and A.~Sardini.
\newblock {Aggressivity and dominance in river crab potamon fluviatile
  (herbst)}.
\newblock \emph{Monitore Zoologico Italiano - Italian Journal of Zoology},
  1971.
\newblock ISSN 00269786.
\newblock \doi{10.1080/00269786.1971.10736174}.

\bibitem[Wang et~al.(2011)Wang, Zhu, Zhu, Zhang, Lin, and Hu]{Wang2011}
F.~Wang, J.~Zhu, H.~Zhu, Q.~Zhang, Z.~Lin, and H.~Hu.
\newblock {Bidirectional control of social hierarchy by synaptic efficacy in
  medial prefrontal cortex}.
\newblock \emph{Science}, 2011.
\newblock ISSN 10959203.
\newblock \doi{10.1126/science.1209951}.

\end{thebibliography}

\newpage
\section*{Appendix: Computational method to find the optimal social structure} \label{sec:appendixa}
\begin{enumerate}
    \item Set N = given population, m = given memory ability.
    \item Create the memory profile matrix base and the strategy profile matrix base, they are N-by-N null matrices.
    \item Digitize the strategies (we use: hawk = 3, dove = 1, mix = 2), to be used in the strategy profile matrix.
    \item Generate an arbitrary memory profile matrix under the given memory ability.
    \item Generate an arbitrary strategy profile matrix at equilibrium (the off-diagonal pairs must sum to 4, meaning players must play $H-D$, $D-H$, or $M-M$ at equilibrium).
    \item Check whether the strategy profile can be supported by the memory profile. The strategy profile can be supported by the memory profile if the agents play the same strategy against all non-memorized opponents:
    \begin{enumerate}
        \item Locate all 0-value entries in the first row in the memory profile matrix. This reflects all non-memorized opponents for the first agent.
        \item Find all corresponding entries in the strategy profile matrix to the entries found in (a).
        \item If all entries found in (b) have the same value, then this agent's strategy profile can be supported by its memory profile.
        \item Repeat (a) to (c) for every agents. The strategy profile matrix is supported by the memory profile matrix if all agents' strategy profiles can be supported by their memory profiles.
          \end{enumerate}
    \item If the strategy profile can be supported by the memory profile, then it is an equilibrium. Count the number of strategy $M$ by counting entries with value=2 in the strategy profile matrix. Divide by 2 will give us the number of pairs that play $M-M$.
    \item Repeat 5-7 for all possible strategy profile matrices at equilibrium.
    \item Repeat 4-8 for all possible memory profile matrices under the given memory ability.
    \item The social structure that induces the least number of $M-M$ pairs is the optimal social structure. If there are multiple, compare their memory profile matrices. The one with more zeroes (fewer ones) is the superior one in terms of fitness, because of lower memory usage.
    \end{enumerate}

\end{document}